\documentclass[final,5p,twocolumn]{elsarticle}

\makeatletter
\def\ps@pprintTitle{%
 \let\@oddhead\@empty
 \let\@evenhead\@empty
 \def\@oddfoot{\centerline{\thepage}}%
 \let\@evenfoot\@oddfoot}
\makeatother

\usepackage{amsthm, amsmath, graphicx}

\newtheorem{theorem}{Theorem}
\usepackage{array}
\usepackage{multirow}
\usepackage{algorithm}
\usepackage{algpseudocode}
\usepackage{color}

\usepackage{epstopdf}
\usepackage{amssymb}

\DeclareMathOperator*{\argmin}{arg\,min}
\DeclareMathOperator*{\argmax}{arg\,max}

\usepackage{algpseudocode}
\algnewcommand\algorithmicto{to}
\algrenewtext{For}[3]%
{\algorithmicfor\ $ #1 = #2$ \algorithmicto\ $#3$ \algorithmicdo}
\newtheorem{lemma}{Lemma}

\def\mydiv{{\big/}}
\def\Kmax{K_{\text{max}}}
\def\x{\mathbf{x}}
\def\y{\mathbf{y}}
\def\res{\mathbf{r}}
\def\AOMP{{A$^\star$OMP}}
\def\AOMPK{{\AOMP$_K$}}
\def\AOMPe{{\AOMP$_e$}}
\def\T{\mathcal{T}}
\def\bphi{\boldsymbol{\phi}}

\newcommand{\subto}{\hspace{4mm}\mbox{subject to\hspace{4mm}}}

\newcolumntype{x}[1]{%
>{\centering\hspace{0pt}}p{#1}}%

\algrenewcommand\algorithmicindent{1.25em}

\makeatletter
\def\blfootnote{\gdef\@thefnmark{}\@footnotetext}
\makeatother

\journal{Signal Processing}

\begin{document}

\begin{frontmatter}
\title{Improving \AOMP: Theoretical and Empirical Analyses\\With a Novel Dynamic Cost Model}

\cortext[cor]{Corresponding author}
\author[bte,sabanci]{Nazim Burak Karahanoglu\corref{cor}}
\ead{burak.karahanoglu@tubitak.gov.tr}

\author[sabanci]{Hakan Erdogan}
\ead{haerdogan@sabanciuniv.edu}

\address[bte]{TUBITAK BILGEM, Kocaeli 41470, Turkey}
\address[sabanci]{Department of Electronics Engineering, Sabanci University, Istanbul 34956, Turkey}

\begin{abstract}
Best-first search has been recently utilized for compressed sensing (CS) by the A$^\star$ orthogonal matching pursuit (\AOMP) algorithm. 
In this work, we concentrate on theoretical and empirical analyses of \AOMP. 
We present a restricted isometry property (RIP) based general condition for exact recovery of sparse signals via \AOMP. 
In addition, we develop online guarantees which promise improved recovery performance with the residue-based termination instead of the sparsity-based one. 
We demonstrate the recovery capabilities of {\AOMP} with extensive recovery simulations using the adaptive-multiplicative (AMul) cost model, which effectively compensates for the path length differences in the search tree. 
The presented results, involving phase transitions for different nonzero element distributions as well as recovery rates and average error, reveal not only the superior recovery accuracy of \AOMP, but also the improvements with the residue-based termination and the AMul cost model. 
Comparison of the run times indicate the speed up by the AMul cost model. 
We also demonstrate a hybrid of OMP and {\AOMP} to accelerate the search further.  
Finally, we run {\AOMP} on a sparse image to illustrate its recovery performance for  more realistic coefficient distributions. 
\end{abstract}
\begin{keyword}
compressed sensing, A$^\star$ orthogonal matching pursuit, restricted isometry property, adaptive-multiplicative cost model
\end{keyword}

\end{frontmatter}

\section{Introduction}
\label{sec:Intro}

\blfootnote{Appeared in Signal Processing, Volume 118, Pages 62-74, January 2016.}

A$^\star$ orthogonal matching pursuit (\AOMP) \cite{Karahanoglu:AOMPfull} aims at combination of best-first tree search
with the orthogonal matching pursuit (OMP) algorithm \cite{Pati:OMP} for the compressed sensing problem.
It incorporates the A$^\star$ search technique \cite{Hart:FBHDMCP, Jelinek:SMSP} to increase the efficiency of the search over the tree representing the hypotheses space consisting of sparse candidates.
Dedicated cost models have been proposed in order to make the search intelligently guided and tractable.
The empirical investigation in \cite{Karahanoglu:AOMPfull} implies significant recovery improvements over conventional compressed sensing methods.
The recently introduced adaptive-multiplicative cost model has led to further speed and accuracy improvements in the preliminary empirical findings of \cite{AOMP_EUSIPCO}.

This paper addresses a detailed theoretical and experimental study of A*OMP algorithm particularly when the number of elements in the sparse solution are not restricted to the sparsity level $K$ of the underlying signal.
This is obtained by enforcing a residue-based termination criterion.
We analyse the theoretical performance of the algorithm using the restricted isometry property.
The analyses cover two possible cases for the length of the returned solution, namely when it is restricted to $K$ nonzero elements, and when more than $K$ elements are allowed.
In addition, the impacts of the adaptive-multiplicative cost model and the residue-based termination criterion on the recovery speed and accuracy are evaluated via comprehensive simulations involving different signal statistics, phase transitions and images in comparison to conventional recovery algorithms.

\subsection{Compressed Sensing}
\label{sec:CS}

The fundamental goal of compressed sensing (CS) is to unify data acquisition and compression by observing a lower dimensional vector
$\y=\mathbf{\Phi}\x$
instead of the signal $\x$, where  $\mathbf{\Phi}\in{\mathbb{R}}^{M\times{N}}$ is the (generally random) measurement matrix\footnote{A more general model involves a structured dictionary $\mathbf{\Psi}$ for sparse representation of $\x$, i.e.,  $\x = \mathbf{\Psi}\mathbf{z}$ where $\mathbf{z}$ is sparse and $\x$ is not. In this more general case, observation model can be written as $\y = \mathbf{\Phi}\mathbf{\Psi}\mathbf{z}$. For simplicity, we omit $\mathbf{\Psi}$ and treat $\x$ as sparse.},
$\y\in{\mathbb{R}}^{M}$, and $\x\in{\mathbb{R}}^{N}$.
The dimensionality reduction follows $M<N$, as a result of which $\x$ cannot be directly solved back from $\y$. Alternatively, assuming $\x$ is $K$-sparse (i.e., it has at most $K$ nonzero components), or compressible, $\x$ can be recovered under certain conditions by solving
\begin{equation}\label{Eq:L0Minimization}
\min\|\x\|_{0} \subto \y=\mathbf{\Phi}\x
\end{equation}
where $\|\x\|_0$ denotes the number of nonzero elements in $\x$.

As the direct solution of (\ref{Eq:L0Minimization}) is intractable, approximate solutions have emerged in the CS literature.
Convex optimization algorithms \cite{Chen:BP, Candes:DecLP, Donoho:CS, Tropp2006589} relax (\ref{Eq:L0Minimization}) by replacing $\|\x\|_{0}$ with its closest convex approximation $\|\x\|_{1}$.
Greedy algorithms \cite{Pati:OMP, Dai:SP, Blumensath:IHT2, FBP_DSP, Sundman2014298, Tropp:SOMP2} provide simple and approximate solutions via iterative residue minimization. Other recovery schemes include Bayesian methods \cite{Ji2008, Babacan2010}, nonconvex approaches \cite{Mohimani:SL0, Mohammadi201442, Montefusco20132636, Ince2013338}, iterative reweighted methods \cite{Wang:ISD, Candes:EncSparREw,Fang2014201}, etc.

\subsection{A$^\star$ Orthogonal Matching Pursuit}
\label{sec:AStar}

OMP is among the most acknowledged greedy algorithms for sparse recovery \cite{Tropp:OMP, Davenport:AnalysisOMP, Wang:OMP_analysis}. It aims at iterative detection of the support, i.e., the set of indices corresponding to nonzero coefficients, of $\x$. At each iteration, OMP identifies the best match to the residue of $\y$ among the atoms, i.e., columns of $\mathbf{\Phi}$, by choosing the index of the atom with maximum correlation. That is, OMP is structurally based on continuous expansion of a single hypothesis, represented by a single iteratively-expanded support estimate, or path.

On the other hand, simultaneous evaluation of multiple hypotheses for the sparse support may improve recovery over the single-path algorithms like OMP. Multiple hypotheses may be represented by a search tree, and the recovery problem can be efficiently solved by sophisticated best-first search techniques. To promote this idea, the authors have introduced the {\AOMP} algorithm \cite{Karahanoglu:AOMPfull, AOMP_EUSIPCO} which is an iterative semi-greedy approach that utilizes A$^\star$ search on a multiple hypotheses search tree in order to find an approximation of (\ref{Eq:L0Minimization}). The nodes of the tree contain indices of the selected atoms, and the paths represent the candidate support sets for $\x$. \figurename~\ref{fig:OMP_AStar} illustrates the evaluation of such a tree in comparison to OMP. At each iteration, {\AOMP} first selects the best path with the minimum cost criterion which depends on the $\ell_2$ norm of the path residue. Then, the best path is expanded by exploring $B$ of its child nodes with maximum correlation to the path residue. That is, $B$ new candidate support sets are appended to the tree. The filled nodes in \figurename~\ref{fig:OMP_AStar} indicate the child nodes explored per step, which are referred to as $\Delta \mathcal{T}$ in the rest.

In \figurename~\ref{fig:OMP_AStar}, although the OMP solution is among the hypotheses in the tree, {\AOMP} returns a different support set. In fact, if OMP were successful, {\AOMP} would also return the same solution. This typical example of OMP failure, where the best-first search identifies the true solution by simultaneous evaluation of multiple hypotheses, illustrates how the multiple path strategy improves the recovery.

Despite this simple illustration, combining the A$^\star$ search with OMP is not straightforward. It necessitates properly defined cost models which enable the A$^\star$ search to perform the stage-wise residue minimization in an intelligent manner, and effective pruning techniques which make the algorithm tractable. Various structures are introduced in \cite{Karahanoglu:AOMPfull} and \cite{AOMP_EUSIPCO} for the cost model, which is vital for the comparison of paths with different lengths. Pruning strategies, which enable a complexity-accuracy trade-off together with the cost model, are detailed in \cite{Karahanoglu:AOMPfull}. Below, we provide a summary of \AOMP, and refer the interested reader to \cite{Karahanoglu:AOMPfull} and \cite{AOMP_EUSIPCO} for the details.

\begin{figure*}[!t]
  \centering
\centerline{\includegraphics[width = 0.7\linewidth]{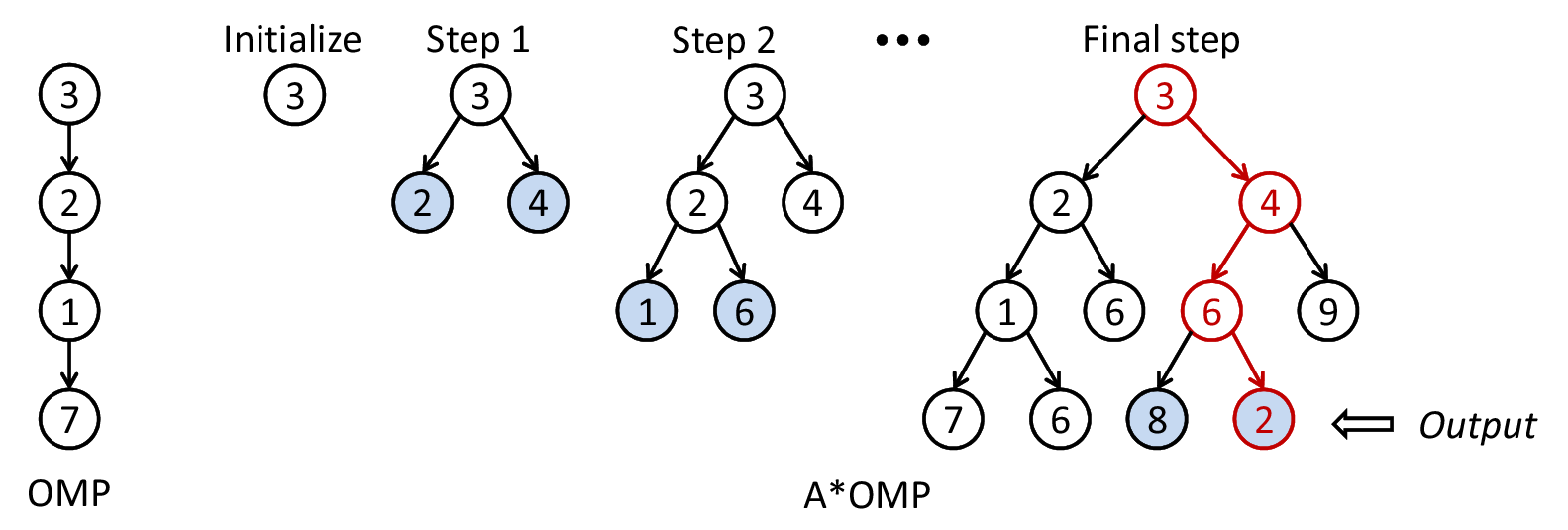}}
\caption{OMP vs. evaluation of {\AOMP} search tree.}
\label{fig:OMP_AStar}
\end{figure*}

\subsubsection{Notation}
\label{sec:Notation}

Before the summary, we clarify the notation in this paper. We define $\mathcal{S}$ as the set of all paths in the search tree. $\mathcal{T}$ is the true support of $\x$.
$\mathcal{T}^i$, $\res^i$, $l^i$ and $f(\mathcal{T}^i)$ denote the support estimate, residue, length and cost of the $i$th path, respectively.
$\hat{\x}^i$ is the estimate of $\x$ given by the $i$th path.
The best path at a certain step is referred to as $b$.
As mentioned above, $\Delta \mathcal{T}$ represents the set of indices selected during the expansion of $b$, i.e., the indices of the $B$ largest magnitude elements in $\mathbf{\Phi}^*\res^b$, where $\mathbf{\Phi}^*$ denotes the conjugate of $\mathbf{\Phi}$.
$\bphi_j$ is the $j$th column of $\mathbf{\Phi}$.
$\mathbf{\Phi}_{\mathcal{J}}$ denotes the matrix composed of the columns of $\mathbf{\Phi}$ indexed by the set $\mathcal{J}$. Similarly, $\x_{\mathcal{J}}$ is the vector of the elements of $\x$ indexed by $\mathcal{J}$.
$\Kmax$ is the maximum number of allowable nodes along a path in the {\AOMP} search tree.
We say that path $i$ \textit{complete} if $l^i = \Kmax$.

\subsubsection{Brief Overview of \AOMP}
\label{sec:AStarBrief}

{\AOMP} initializes the search tree with $I$ paths of a single node each. These nodes represent the indices of the $I$ largest magnitude elements in  $\mathbf{\Phi}^*\y$.
At each iteration, the algorithm first selects the best path $b$ among the incomplete paths in the search tree with minimum cost.
Then, $\Delta \mathcal{T}$ is chosen as the indices of the $B$ largest magnitude elements in $\mathbf{\Phi}^*\res^b$.
This implies $B$ candidate paths, each of which expands $b$ with a single index in $\Delta \mathcal{T}$.
Each candidate path is opened unless an equivalent path has been explored before (\textit{equivalent path pruning} \cite{Karahanoglu:AOMPfull}).
For each new path $i$, $\res^i$ is given by the projection error of $\y$ onto $\mathbf{\Phi}_{\mathcal{T}_i}$, and the cost $f(\mathcal{T}_i)$ is computed as a function of $\res^i$. Finally, all but the best $P$ paths with minimum cost are pruned  (\textit{tree size pruning} \cite{Karahanoglu:AOMPfull}).
Selection and expansion of the best path are repeated until either some path $i$ satisfies $\|\res^i\|_2 \leq \varepsilon\|\y\|_2$, or all $P$ paths are complete. The pseudo-code for {\AOMP} is given in Algorithm~\ref{alg:AOMP}.

\begin{algorithm}[!t]
\caption{A$^\star$ ORTHOGONAL MATCHING PURSUIT}
\label{alg:AOMP}
\small
\begin{algorithmic}[1]
\State \textbf{Input:} $\mathbf{\Phi}$, $\y$
\State \textbf{Define:} $P$,  $I$, $B$, $\Kmax$, $\varepsilon$, $\{\alpha_{\text{Mul}}\; \text{or} \; \alpha_{\text{AMul}}\}$
\State \textbf{Initialize:} $\mathcal{T}^i = \emptyset$, $r^i = \mathbf{y}$ $\forall{i=1,2,...,P}$, $b =\{1\}$
\State $\Delta \mathcal{T} = \argmax\limits_{\mathcal{J}, |\mathcal{J}|=I} \sum_{j \in \mathcal{J}}|\langle\bphi_j,\y\rangle|$
\For{i}{1}{I}  \Comment{$I$ paths of length $1$}
    \State $\mathcal{T}^i = \{\text{$i$th index in $\Delta \mathcal{T}$}\}$, $\res^i = \y -  \langle\y,\bphi_{\mathcal{T}^i}\rangle\bphi_{\mathcal{T}^i}$
\EndFor
\While{$b \neq \emptyset$}
    \State $\Delta \mathcal{T} = \argmax\limits_{\mathcal{J}, |\mathcal{J}|=B} \sum_{j \in \mathcal{J}}|\langle\bphi_j,\res^b\rangle|$   \Comment{$B$ children of $b$}
    \State $\widetilde{\mathcal{T}} = \mathcal{T}^b$
    \State $p = b$   \Comment{first to be replaced}
    \For{i}{1}{B} \Comment{expansion}
        \State $\widehat{\mathcal{T}} =\widetilde{\mathcal{T}} \cup \{\text{$i$th index in $\Delta \mathcal{T}$}\}$ \Comment{candidate path}
        \State $\mathbf{z} = \argmin\limits_{\hat{\mathbf{z}}} \| \y- \mathbf{\Phi}_{\widehat{\mathcal{T}}} \hat{\mathbf{z}} \|_{2}$ \Comment{orthogonal projection}
        \State $\hat{\res} = \y - \mathbf{\Phi}_{\widehat{\mathcal{T}}} \mathbf{z}$ \Comment{update residue}
        \If{$(\|\hat{\res}\|_2 \leq \varepsilon\|\y\|_2)$} \Comment{check residue}
            \State \textbf{return} $\widehat{\mathcal{T}}$ \Comment{terminate}
        \EndIf
        \If{$f(\widehat{\mathcal{T}}) < f(\mathcal{T}^p)$ and $(\widehat{\mathcal{T}} \notin \mathcal{S})$} \Comment{pruning}
            \State $\mathcal{T}^p = \widehat{\mathcal{T}}$, $\res^p = \hat{\res}$
        \EndIf
        \State $p = \argmax\limits_{i \in 1,2,...,P} f(\mathcal{T}^i)$ \Comment{worst path(replaced next)}
    \EndFor
    \State $b = \argmin\limits_{i \in 1,2,...,P,\;l^i<\Kmax} f(\mathcal{T}^i)$   \Comment{best incomplete path}
\EndWhile
\State $b = \argmin\limits_{i \in 1,2,...,P} f(\mathcal{T}^i)$  \Comment{best (complete) path}
\State \textbf{return} $\mathcal{T}^b$ 
\end{algorithmic}
\end{algorithm}

\subsubsection{Termination Criteria}
\label{sec:TerCriteria}

{\setlength{\parskip}{0cm}
The outline presented above is slightly different from the introduction of {\AOMP} in \cite{Karahanoglu:AOMPfull} and \cite{AOMP_EUSIPCO}.
The best path selection and termination mechanisms of {\AOMP} are modified in order to improve the theoretical guarantees of the algorithm.
In \cite{Karahanoglu:AOMPfull} and \cite{AOMP_EUSIPCO}, the best path is selected among \emph{all} paths in the tree.
Accordingly, the search terminates either when

\hspace{4mm}(\emph{i}) the best path is complete ($l^b==\Kmax$), or

\hspace{4mm}(\emph{ii}) the residue is small enough ($||\res||_2<\epsilon||\y||_2$).

\noindent Actually, (\emph{i}) appears here as a consequence of the best path choice involving complete paths in addition to the incomplete ones.
Since the best path is chosen using the cost function, this termination criterion depends on the cost model.
To guarantee exact recovery with this scheme, the cost model should exhibit some sense of optimality, i.e., it should assign potentially true paths lower costs than false complete paths.
This is necessary to ensure that some false complete path does not become the best path.
However, optimality of the cost model is analytically hard to guarantee.
This becomes an obstacle for obtaining theoretical guarantees independent of the cost models.}

{\setlength{\parskip}{0cm}
To overcome this problem, the dependency of the termination on the cost function, i.e., termination criterion (\emph{i}), should be removed.
For this purpose, we modify {\AOMP} as follows:

\hspace{4mm}(\emph{i}) The best path is selected among the \emph{incomplete} paths in the tree (Algorithm \ref{alg:AOMP}, line 24).

\hspace{4mm}(\emph{ii}) The search terminates when the residue is small enough (Algorithm \ref{alg:AOMP}, line 16-17).

\noindent To ensure termination, the best complete path is returned as the solution only when all paths are complete, but none of them satisfies the termination criterion on the residue\footnote{Note that this already indicates a recovery failure.} (Algorithm \ref{alg:AOMP}, line 26-27).
These modifications are listed in Table~\ref{Table:TerCrit}.
In this structure, termination does not directly rely on the cost model, hence stronger exact recovery guarantees may be obtained.
The results in this paper are obtained with the modified version of the algorithm.
Note that, based on the authors' experience, this modification does not have a significant effect on the empirical performance of the algorithm. Yet, it is critical for the theoretical analysis.
The rest of this manuscript concentrates on the modified version of the algorithm without explicit referral.}

\begin{table}[!t]
\caption{Comparison of {\AOMP} mechanisms}
\centering
\setlength\extrarowheight{2pt}
{\footnotesize{
\begin{tabular}{m{0.18\linewidth} | m{0.33\linewidth} | m{0.33\linewidth} }  \hline \hline
& Initial version \cite{Karahanoglu:AOMPfull, AOMP_EUSIPCO}   	& Modified version	\tabularnewline \hline
Best path

selection & among \emph{all} paths & among \emph{incomplete}

 paths	\tabularnewline \hline
Termination Criteria&
--   The best path is complete.

--   The residue is small enough. & The residue is small

enough.	\tabularnewline \hline\hline
\end{tabular}
}}
\label{Table:TerCrit}
\end{table}

The termination parameters in Algorithm~\ref{alg:AOMP}, $\Kmax$ and  $\varepsilon$, can be adjusted for different termination behaviour.
 In \cite{Karahanoglu:AOMPfull}, each path is limited to $K$ nodes, i.e., $\Kmax=K$.
We call this sparsity-based termination, and denote by {\AOMPK}.
Another alternative is the residue-based termination such as in \cite{AOMP_EUSIPCO}, where more than $K$ nodes are allowed along a path by setting $\Kmax>K$ and $\varepsilon$ is selected small enough based on the noise level.
This version is referred to as {\AOMPe}.
Preliminary results in \cite{AOMP_EUSIPCO} indicate that {\AOMPe} yields not only better recovery but also faster termination than {\AOMPK}.
With the flexibility on choosing $\Kmax$ and $\varepsilon$, we apply both termination criteria in Algorithm~\ref{alg:AOMP}.

\subsubsection{Cost Models}
\label{sec:CostModels}

To select the best path, {\AOMP} should compare the costs of paths with different lengths. This necessitates proper cost models which can compensate for the differences in path lengths. Some novel models have been proposed in \cite{Karahanoglu:AOMPfull} and \cite{AOMP_EUSIPCO}. In this work, we employ the multiplicative (Mul) \cite{Karahanoglu:AOMPfull} and adaptive-multiplicative (AMul) \cite{AOMP_EUSIPCO} models, following their superior recovery capabilities demonstrated in \cite{Karahanoglu:AOMPfull} and \cite{AOMP_EUSIPCO}.

The Mul cost model relies on the expectation that unexplored nodes decrease $\left\| \res^i \right\|_2$ by a constant rate $\alpha_{\text{Mul}} \in (0,1)$:
\begin{equation}
    f_{\text{Mul}}(\mathcal{T}^i) = \alpha_{\text{Mul}}^{\Kmax-l^i} \left\| \res^i \right\|_2. \nonumber
\end{equation}
Note that, we replace $K$ in \cite{Karahanoglu:AOMPfull} with $\Kmax$ to allow for different termination criteria. \cite{Karahanoglu:AOMPfull} demonstrates that decreasing $\alpha_{\text{Mul}}$ improves recovery accuracy, while the search gets slower.

The AMul model is a dynamic extension of the Mul model:
\begin{equation} \label{Eq:AMul_CM}
    f_{\text{AMul}}(\mathcal{T}^i) = \left(\alpha_{\text{AMul}} \frac{\left\| \res^i_{l^i} \right\|_2}{\left\| \res^i_{l^i-1} \right\|_2}\right)^{\Kmax-l^i} \left\| \res^i_{l^i} \right\|_2
\end{equation}
where $\res^i_{l}$ denotes the residue after the first $l$ nodes of the path $i$, and $\alpha_{\text{AMul}}\in (0,1]$ is the cost model parameter.

The AMul cost model relies on the following assumption: each unexplored node would reduce $\left\| \res^i \right\|_2$ by a rate proportional to the decay occurred during the last expansion of the path $i$. This rate is modeled by the auxiliary term $\alpha_{\text{AMul}} \left\| \res^i_{l^i} \right\|_2\mydiv\left\| \res^i_{l^i-1} \right\|_2$, and the exponent $\Kmax-l^i$ extends this to all unexplored nodes along path $i$.
The motivation is intuitive: since the search is expected to select nodes with descending correlation to $\y$, a node is expected to reduce $\left\| \res^i \right\|_2$ less than its ancestors do.
Note that this condition may be violated for a particular node. However, the auxiliary term is mostly computed over a number of nodes instead of a single one. Hence, it is practically sufficient if this assumption holds for groups of nodes.
Moreover, the tree usually contains multiple paths which may lead to the correct solution.
The fact that some of these paths violate this assumption does not actually harm the recovery.
This behavior is similar to the other cost models in \cite{Karahanoglu:AOMPfull}. The empirical results in  Section~\ref{sec:results} indicate that these cost models are useful in practice.

The adaptive structure of the AMul model allows for a larger $\alpha$ than the Mul model.
This reduces the auxiliary term. Consequently, the search favors longer paths, explores fewer nodes and terminates faster as demonstrated in Section~\ref{sec:results}.

In the rest, we identify the cost model employed by {\AOMP} with an appropriate prefix. AMul-{\AOMP} and Mul-{\AOMP} denote the use of AMul and Mul cost models, respectively.

\subsubsection{Relations of {\AOMP} to Recent Proposals}
\label{sec:Relations}

One of the first and trivial combinations of the tree search with matching pursuit type algorithms has been suggested in \cite{Cotter:TSBOMP} in 2001.
Two strategies have been considered for exploring a tree with branching factor $K$, that is where each node has $K$ children only\footnote{In this context, $K$ is not related to the sparsity level of $\x$ as before. We denote the branching factor as $K$ in order to be consistent with \cite{Cotter:TSBOMP}.}.
MP:K has a depth-first nature, that is the candidate paths are explored one by one.
The algorithm first explores a complete path up to the maximum depth.
If this path does not yield the desired solution, the tree is backtracked and  other candidates are explored sequentially until the solution is found.
The other variant, MP:M-L, is based on breadth-first search.
It processes all leaf nodes at a certain depth at once by exploring $K$ children of each leaf and keeps the best $M$ among the new candidates.
The process is repeated until tree depth becomes $L$, and the path with the lowest residual is returned.

This idea has recently been revisited in \cite{Kwon:MMP}, where the algorithm is referred to as multipath matching pursuit (MMP).
As in \cite{Cotter:TSBOMP}, breadth-first (MMP-BF) and depth-first (MMP-DF) strategies have been evaluated to explore a search tree with branching factor $L$.
For tractability, MMP-DF sets a limit on the number of sequentially explored paths.
As a novel contribution, \cite{Kwon:MMP} provides RIP-based theoretical guarantees for MMP. Note that these guarantees are applicable when the number of new paths per level or the number of explored paths are not limited, i.e., when no pruning is applied.

Though both ideas are based on exploring a search tree, these algorithms are fundamentally different than {\AOMP}, where the tree search is guided by adaptive cost models  in an intelligent manner.  Instead, MMP-BF and MMP-DF are rather unsophisticated techniques  where tree search follows a predefined order.
We compare {\AOMP} and MMP via recovery simulations in Section~\ref{sec:results}. MMP-DF is chosen among the two variants, since it is referred to as the practical one in \cite{Kwon:MMP}.

\subsection{Outline and Contributions}
\label{sec:Outline}

The manuscript at hand concentrates on detailed analyses of the sparse signal recovery performance of \AOMP. Particularly, we concentrate on the variant AMul-{\AOMPe} which extends the general form in \cite{Karahanoglu:AOMPfull} by the novel AMul cost model from \cite{AOMP_EUSIPCO} and residue-based termination. We present new theoretical and empirical results to demonstrate the superiority of this variant not only over the {\AOMP} variants in \cite{Karahanoglu:AOMPfull}, but also over some conventional sparse recovery methods. Note that AMul-{\AOMPe} has only been preliminarily tested in \cite{AOMP_EUSIPCO} by a set of limited simulations, which are far away from providing enough evidence to generalize its performance. This manuscript presents a detailed empirical investigation of AMul-{\AOMPe}, without which the performance analyses would not be complete.
These simulations significantly enrich the findings of \cite{AOMP_EUSIPCO} by previously unpublished results which include phase transition comparisons for different signal statistics, demonstration on an image, a faster hybrid approach and optimality analyses. The results reveal not only the superior recovery accuracy of AMul-{\AOMPe}, but also the improvements in the speed of the algorithm due to the residue-based termination and the AMul cost model.

On the other hand, our theoretical findings not only include RIP-based exact recovery guarantees for exact recovery of sparse signals via {\AOMPK} and {\AOMPe}, but also provide means for comparison of different termination criteria.
The former states RIP conditions for the exact recovery of sparse signals from noise-free measurements, while the latter addresses the recovery improvements when the residue-based termination is employed instead of the sparsity-based one.
For the analyses, we employ a method similar to the analyses of the OMP algorithm in \cite{Wang:OMP_analysis} and  \cite{Karahanoglu:OMPe}.
In Section~\ref{Sec:SingIter}, we develop a RIP condition for the success of a single {\AOMP} iteration, which forms a basis for the following theoretical analyses.
In Section~\ref{Sec:Analysis_AOMPK}, we derive a general recovery condition for exact recovery via {\AOMPK}.
As intuitively expected, this condition is less restrictive than the $K$-step OMP recovery condition \cite{Wang:OMP_analysis}.
Section~\ref{Sec:Analysis_AOMPe} presents a very similar general condition for exact recovery via {\AOMPe}.
In addition, we establish an online recovery condition for exact recovery of a signal with {\AOMPe}.
Section~\ref{Sec:Analysis_Comp} compares the general and online recovery conditions, clarifying that the latter is less restrictive. This suggests that {\AOMPe} possesses stronger recovery capabilities than {\AOMPK}.

Section~\ref{sec:results} compares the recovery accuracy of AMul-{\AOMPe} to other {\AOMP} variants, basis pursuit (BP) \cite{Chen:BP}, subspace pursuit (SP) \cite{Dai:SP}, OMP \cite{Pati:OMP}, iterative hard thresholding (IHT) \cite{Blumensath:IHT2}, iterative support detection (ISD) \cite{Wang:ISD}, smoothed $\ell_0$ (SL0) \cite{Mohimani:SL0}, MMP \cite{Kwon:MMP}, and forward-backward pursuit (FBP) \cite{FBP_DSP}.
The main contribution is the phase transitions which are obtained by computationally expensive experiments for different signal types.
These generalize the strong recovery capability of AMul-{\AOMPe} over a wide range of $N$, $M$ and $K$.
We investigate the recovery rates and average recovery error as well.
Run times illustrate the acceleration with the AMul cost model and the residue-based termination.
An hybrid of {\AOMPe} and OMP demonstrates how the recovery speed can be improved without losing the accuracy.
The sparse image recovery problem represents a more realistic coefficient distribution than the other artificial examples.

\section{Theoretical Analysis of \AOMP}
\label{Sec:Analysis}

In this section, we  develop theoretical guarantees  for signal recovery with {\AOMP}. We first visit the restricted isometry property and then provide some related preliminary lemmas. Then, we concentrate recovery with {\AOMP}.

\subsection{Restricted Isometry Property}
Restricted isometry property (RIP) \cite{Candes:DecLP} provides an important means for theoretical guarantees in sparse recovery  problems.
A matrix $\mathbf{\Phi}$ is said to satisfy the $L$-RIP if there exists a restricted isometry constant (RIC) $\delta_L \in (0,1)$ satisfying
\begin{equation}\label{Eq:RIP}
    (1-\delta_L)\|\x\|_2^2 \leq \|\mathbf{\Phi}\x\|_2^2 \leq (1+\delta_L)\|\x\|_2^2, \:\: \forall \x{:}\|\x\|_0 \leq L. \nonumber
\end{equation}

Some random matrices, such as Gaussian or Bernoulli matrices, satisfy the $L$-RIP with high probabilities if $L$, $M$ and $N$ satisfy some specific conditions \cite{Rudelson:SparseRec, Candes:NOptRec}. Exploiting this property, RIP has been utilized to obtain recovery guarantees for sparse recovery algorithms \cite{Candes:DecLP, Candes:NOptRec, Karahanoglu:OMPe, Xu20132653, Wang2014188}.

\subsection{Preliminaries}

We now present some preliminary lemmas based on RIP:

\begin{lemma}[Monotonicity of the RIC, \cite{Dai:SP}]\label{lemma_monotonicity}
Let $R$ and $S$ be positive integers such that $R > S$. Then,
$\delta_R \geq \delta_S$.
\end{lemma}
\begin{lemma}[Lemma 2, \cite{Wang:OMP_analysis}] \label{lemma_SP1}
Let $\mathcal{I} \subset \{1,2,\cdots,N\}$ and $|\mathcal{I}|$ denote the cardinality of $\mathcal{I}$.
For any arbitrary vector $\mathbf{z} \in \mathbb{R}^{|\mathcal{I}|}$, RIP directly leads to
\begin{equation}
(1-\delta_{|\mathcal{I}|})\|\mathbf{z}\|_2 \leq \|\mathbf{\Phi}_{\mathcal{I}}^*\mathbf{\Phi}_{\mathcal{I}}\mathbf{z}\|_2 \leq (1+\delta_{|\mathcal{I}|})\|\mathbf{z}\|_2. \nonumber
\end{equation}

\end{lemma}
\begin{lemma}[Lemma 1, \cite{Dai:SP}] Let $\mathcal{I},\mathcal{J} \subset \{1,2,\cdots,N\}$ such that $\mathcal{I} \cap \mathcal{J} = \emptyset$. For any arbitrary vector $\mathbf{z} \in \mathbb{R}^{|\mathcal{J}|}$
\label{lemma_SP2}
\begin{equation}
\|\mathbf{\Phi}_{\mathcal{I}}^*\mathbf{\Phi}_{\mathcal{J}}\mathbf{z}\|_2 \leq \delta_{|\mathcal{I}|+|\mathcal{J}|}\|\mathbf{z}\|_2. \nonumber
\end{equation}
\end{lemma}
\begin{lemma}
\label{lemma_ExRec2}
Let $K$ and $B$ be positive integers, and $\lceil z \rceil$ denote the smallest integer greater than or equal to $z$. Then,
\begin{equation}
\delta_{K+B} > \frac{\delta_{3\lceil K/ 2 \rceil}}{3}. \nonumber
\end{equation}
\end{lemma}
\begin{proof}
Corollary 2 of \cite{Karahanoglu:OMPe} states that Lemma~\ref{lemma_ExRec2} holds for $B=1$. By Lemma~\ref{lemma_monotonicity}, $\delta_{K+B} \geq \delta_{K+1}$ for $B>1$. Hence, Lemma~\ref{lemma_ExRec2} also holds for $B>1$.
\end{proof}
\begin{lemma}
\label{lemma_ExRec1}
Assume $K \geq (3+2\sqrt{B})^2$. There exists at least one positive integer $n_c<K$ such that
\begin{equation}
\label{Eq:Lemma_ExRec1_1}
 \frac{3\sqrt{B}}{\sqrt{K}+\sqrt{B}} \leq \frac{\sqrt{B}}{\sqrt{K-n_c}+\sqrt{B}}.
\end{equation}
Moreover, $n_c$ values which satisfy (\ref{Eq:Lemma_ExRec1_1}) are bounded by
\begin{equation}
\label{Eq:n_cBound}
 K > n_c \geq \frac{8K+4\sqrt{BK}-4B}{9}.
\end{equation}
\end{lemma}
\begin{proof}
We set $K-n_c = sK$ and replace into (\ref{Eq:Lemma_ExRec1_1}):
\begin{equation}
 \frac{3\sqrt{B}}{\sqrt{K}+\sqrt{B}} \leq \frac{\sqrt{B}}{\sqrt{sK}+\sqrt{B}}. \nonumber
\end{equation}
It can trivially be shown that $s$ is bounded by
\begin{equation}
 0 < s \leq \left( \frac{\sqrt{K}-2\sqrt{B}}{3\sqrt{K}} \right)^2. \nonumber
\end{equation}
Then, we obtain the lower bound for $n_c$ as
\begin{equation}
n_c = (1-s)K \geq \frac{8K+4\sqrt{BK}-4B}{9}.
    \label{Eq:Lemma_ExRec1_2}
\end{equation}
Since $n_c<K$, $sK = K-n_c \geq 1$. This translates as
\begin{equation}
 K {}\geq{} \frac{1}{s} {}\geq{}\left(\frac{3\sqrt{K}}{\sqrt{K}-2\sqrt{B}}\right)^2 \nonumber
\end{equation}
from which we deduce the assumption ${K \geq (3+2\sqrt{B})^2}$. Combining this result with (\ref{Eq:Lemma_ExRec1_2}) completes the proof.
\end{proof}

\subsection{Success Condition of an {\AOMP} Iteration}
\label{Sec:SingIter}

We define the success of an {\AOMP} iteration as $\Delta \mathcal{T}$ containing at least one correct index, i.e., $\Delta \mathcal{T} \cap \{\mathcal{T} - \mathcal{T}^b \} \neq \emptyset$.
The following theorem guarantees the success of an iteration:

\begin{theorem}
\label{Thrm_Main}
Let $n_c = |{\mathcal{T}^b \cap \mathcal{T}}|$ and $n_f = |{\mathcal{T}^b - \mathcal{T}}|$. When $b$ is expanded, at least one index in $\Delta \mathcal{T}$ is in the support of $\x$, i.e., ${\Delta \mathcal{T}\cap \{\mathcal{T} - \mathcal{T}^b \} \neq \emptyset}$ if $\mathbf{\Phi}$ satisfies RIP with
\begin{equation}
\delta_{K+n_f+B} < \min\left(\frac{\sqrt{B}}{\sqrt{K-n_c}+\sqrt{B}},\frac{1}{2}\right).
\label{Eq:AOMP_res1}
\end{equation}
\end{theorem}
\begin{proof}
$\Delta \mathcal{T}$ can be defined as
  \begin{equation}
  \Delta \mathcal{T} = \argmax\limits_{\mathcal{J}, |\mathcal{J}|=B}\left\| \mathbf{\Phi}_{\mathcal{J}}^*\res^b\right\|_2.
    \label{DeltaT}
  \end{equation}

$\res^b$ is the residue from the orthogonal projection of $\y$ onto $\mathbf{\Phi}_{\mathcal{T}^b}$.
Therefore, ${\res^b \perp \mathbf{\Phi}_{\mathcal{T}^b}}$, i.e., $\langle\bphi_i,\res^b\rangle = 0$ if $i \in \mathcal{T}^b.$ Hence,
\begin{equation}\label{Corr_TTl}
\left\|\mathbf{\Phi}_{\mathcal{T}\cup \mathcal{T}^b}^*\res^b\right\|_2^2=\sum_{i \in T\cup \mathcal{T}^b}\left\langle\bphi_i,\res^b\right\rangle^2 \;=\; \sum_{i \in T - \mathcal{T}^b}\left\langle\bphi_i,\res^b\right\rangle^2.
\end{equation}
(\ref{Corr_TTl}) has only $K - n_c$ nonzero terms. Combining (\ref{Corr_TTl}) and (\ref{DeltaT}), we can write
\begin{equation} \label{DeltaT_bound1}
\|\mathbf{\Phi}_{\Delta \mathcal{T}}^*\res^b\|_2 = \max\limits_{\mathcal{J}, |\mathcal{J}|=B} \|\mathbf{\Phi}_{\mathcal{J}}^*\res^b\|_2 \geq c\|\mathbf{\Phi}_{\mathcal{T}\cup \mathcal{T}^b}^*\res^b\|_2
\end{equation}
where the inequality holds since $\mathcal{J}$ maximizes $\|\mathbf{\Phi}_{\mathcal{J}}^*\res^b\|_2$, and
$c$ defines a scaling proportional to the number of nonzero terms:
\begin{equation}
c\triangleq \min \left(\sqrt{\frac{B}{K-n_c}},1\right). \nonumber
\end{equation}

Next, the residue can be written as
\begin{equation}\label{residue}
\res^b = y - \mathbf{\Phi}_{\mathcal{T}^b}\hat{\x}^b_{\mathcal{T}^b} = \mathbf{\Phi}_\mathcal{T}\x_\mathcal{T}- \mathbf{\Phi}_{\mathcal{T}^b}\hat{\x}^b_{\mathcal{T}^b} = \mathbf{\Phi}_{\mathcal{T}\cup \mathcal{T}^b}\mathbf{z}
\end{equation}
where $\mathbf{z}\in \mathbb{R}^{K+n_f}$. Using Lemma~\ref{lemma_SP1}, (\ref{DeltaT_bound1}) and (\ref{residue}), we write
\begin{equation}
\|\mathbf{\Phi}_{\Delta \mathcal{T}}^*\res^b\|_2 \geq c\|\mathbf{\Phi}_{\mathcal{T}\cup \mathcal{T}^b}^*\mathbf{\Phi}_{\mathcal{T}\cup \mathcal{T}^b}\mathbf{z}\|_2 \geq c(1-\delta_{K+n_f})\|\mathbf{z}\|_2. \nonumber
\end{equation}

Now, suppose that ${\Delta \mathcal{T} \cap T = \emptyset}$. Then
\begin{equation} \label{Eq:FalseSelect}
\|\mathbf{\Phi}_{\Delta \mathcal{T}}^*\res^b\|_2 = \|\mathbf{\Phi}_{\Delta \mathcal{T}}^*\mathbf{\Phi}_{\mathcal{T}\cup \mathcal{T}^b}\mathbf{z}\|_2
\leq \delta_{K+n_f+B}\|\mathbf{z}\|_2  \nonumber
\end{equation}
by Lemma~\ref{lemma_SP2}. Clearly, this never occurs if
\begin{equation}
 c(1-\delta_{K+n_f})\|\mathbf{z}\|_2 > \delta_{K+n_f+B}\|\mathbf{z}\|_2 \nonumber
\end{equation}
or equivalently
\begin{equation} \label{Cond1}
 \frac{\delta_{K+n_f+B}}{c} + \delta_{K+n_f} < 1.
\end{equation}
Following Lemma~\ref{lemma_monotonicity}, ${\delta_{K+n_f+B} \geq \delta_{K+n_f}}$. Hence, (\ref{Cond1}) is satisfied when
$\left(\frac{1}{c}+1\right)\delta_{K+n_f+B} < 1$, or equivalently
\begin{equation}
\delta_{K+n_f+B} < \frac{c}{1+c}
= \min\left(\frac{\sqrt{B}}{\sqrt{K-n_c}+\sqrt{B}},\frac{1}{2}\right). \nonumber
\end{equation}
This guarantees that ${\Delta \mathcal{T} \cap \mathcal{T} \neq \emptyset}$.
Moreover, since $\langle\bphi_i,\res^b\rangle = 0$ for all $i \in \mathcal{T}^b$, $\Delta \mathcal{T} \cap \mathcal{T}^b = \emptyset$. Hence, we conclude $\Delta \mathcal{T} \cap \{\mathcal{T} - \mathcal{T}^b\}\neq \emptyset$, that is the {\AOMP} iteration is successful.
\end{proof}

Below, Theorem~\ref{Thrm_Main} is used as a basis for exact recovery. Note that we assume $\sqrt{B} \leq \sqrt{K-n_c}$ in the rest and skip the term $\frac{1}{2}$ in Theorem~\ref{Thrm_Main} for simplicity. This can be justified by the fact that $B$ is chosen small (such as 2 or 3) in practice.

\subsection{Exact Recovery Conditions for {\AOMPK}}
\label{Sec:Analysis_AOMPK}

First, let us introduce some definitions:

\emph{Optimal path:} Path $i$ is said to be \textit{optimal} if $\mathcal{T}^i \subset \mathcal{T}$.

\emph{Optimal pruning:} Pruning is defined as \textit{optimal} if it does not remove all optimal paths from the tree.

\figurename~\ref{fig:optimality_tree} illustrates the optimality notion with a typical example.
Indices from the true support are shown by the filled nodes.
The optimal paths, which contain filled nodes only, are designated with `+'.
Since $\Kmax=K=4$, the paths become complete with four nodes. The nodes which are pruned due to $P=3$ are crossed.
Step 3 exemplifies the optimal pruning. Though an optimal path is pruned, there still remains another optimal one.
Expanding this optimal path in Step 5, the search obtains the true solution represented by the complete and optimal path $\{3,4,6,2\}$.

\begin{figure*}[!t]
  \centering
\centerline{\includegraphics[width = 0.7\linewidth]{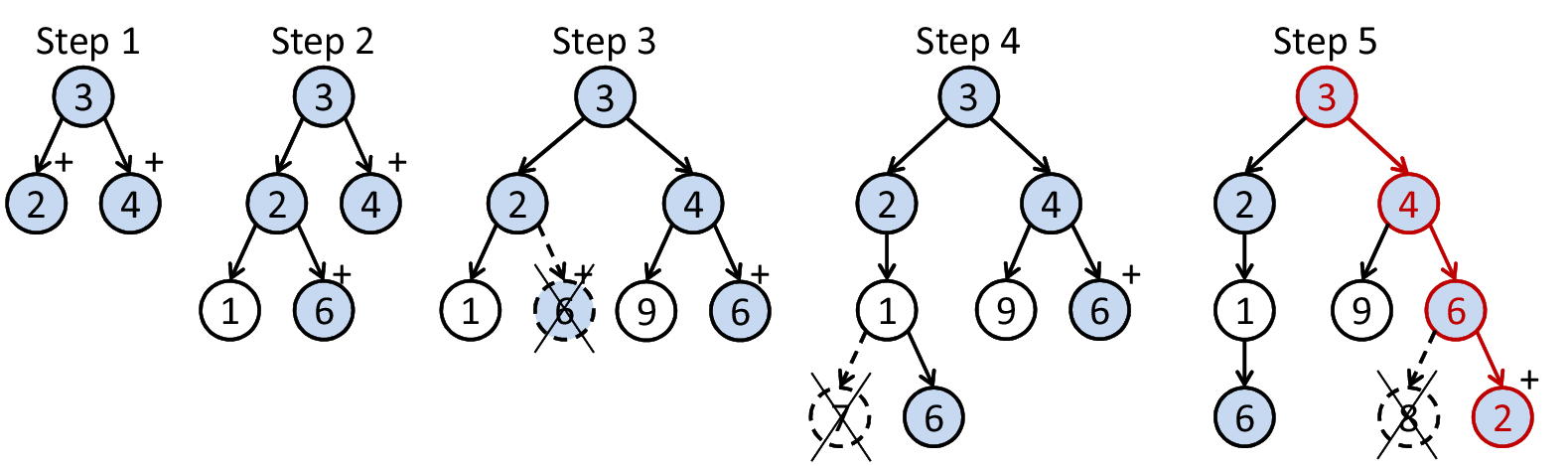}}
\caption{Optimality during the search. The true support is $\{3,4,6,2\}$.}
\label{fig:optimality_tree}
\end{figure*}

Now, we present the exact recovery condition for {\AOMPK}:
\begin{theorem}
\label{Thrm_AOMPK}
Set $\varepsilon = 0$ and $\Kmax = K$. Let $\mathbf{\Phi}$ be full rank and $2K \leq M$ hold\footnote{$2K \leq M$ is a global condition for the uniqueness of all K-sparse solutions. Hence this condition is necessary for any sparse recovery algorithm.}. Assume that pruning is optimal.  Then, {\AOMPK} perfectly recovers all $K$-sparse signals from noise-free measurements if $\mathbf{\Phi}$ satisfies RIP with
\begin{equation}
\delta_{K+B} < \frac{\sqrt{B}}{\sqrt{K}+\sqrt{B}}.
\label{Eq:AOMPK_res1}
\end{equation}
\end{theorem}
\begin{proof}
Let us start with the initialization. With $n_c = n_f = 0$, Theorem~\ref{Thrm_Main} assures success of the first iteration.

Next, consider {\AOMPK} selects an optimal path of length $l$, i.e., $n_c = l$, at some step. By Theorem~\ref{Thrm_Main}, expansion of this path is successful if
\begin{equation}
\delta_{K+B} < \frac{\sqrt{B}}{\sqrt{K-l}+\sqrt{B}}
\end{equation}
which is already satisfied when (\ref{Eq:AOMPK_res1}) holds.

Now, there exists some optimal paths after initialization. Moreover, expanding an optimal path introduces at least one longer optimal path, and by assumption pruning cannot remove all of these. Altogether, these guarantee the existence of at least one optimal path in the tree at any iteration.

On the other hand, the criterion $\varepsilon=0$ requires that the residue should vanish for termination. Since $2K \leq M$ and $\mathbf{\Phi}$ is full rank, the residue may vanish if and only if $\T$ is a subset of the support estimate\footnote{As linearly dependent subsets should contain at least $M+1$ columns of $\mathbf{\Phi}$, any other solution should be at least $(M-K+1)$-sparse.}. Therefore, the search must terminate at a complete optimal path containing $\T$ unless there remain no optimal paths in the search tree. Together with the existence of at least one optimal path, this guarantees exact recovery.
\end{proof}

Note that the condition $\varepsilon = 0$, stated in Theorem~\ref{Thrm_AOMPK} for the sake of theoretical correctness, translates into a very small $\varepsilon$ in practice to account for the numerical computation errors.

We observe that the $K$-step exact recovery condition of OMP, $\delta_{K+1} < 1\mydiv\big(\sqrt{K}+1\big)$, is a special case of Theorem~\ref{Thrm_AOMPK} when $B=I=1$. Moreover, when the bounds for OMP and {\AOMPK} are compared, (\ref{Eq:AOMPK_res1}) is clearly less restrictive, which explains the improved recovery accuracy of {\AOMPK}.

Theorem~\ref{Thrm_AOMPK} is closely related to the theoretical analysis of MMP in \cite{Kwon:MMP}. It can be trivially shown that Theorem~\ref{Thrm_AOMPK} is also applicable to MMP.
Moreover, it implies a better, i.e., less restricted, recovery condition when compared to \cite{Kwon:MMP}, where the condition is stated as  $\delta_{K+B}<1\mydiv\big(\sqrt{K}+2\sqrt{B}\big)$.

\subsection{Exact Recovery with {\AOMPe}}
\label{Sec:Analysis_AOMPe}

We extend the definitions in the previous section for {\AOMPe} where $\Kmax>K$.
First, note that path $i$ is now complete if $l^i = \Kmax > K$. Next, we introduce the following definitions:

\emph{Potentially-optimal path:} A path is said to be \textit{potentially-optimal (p-optimal)} if $n_f \leq \Kmax-K$. A p-optimal path can be expanded into a superset of $\mathcal{T}$ with at most $\Kmax$ nodes. Note that an optimal path is a special case where  $\Kmax = K$.

\emph{Potentially-optimal pruning:} Pruning is defined as \textit{p-optimal} if it does not remove all p-optimal paths from the search tree.

\figurename~\ref{fig:optimality} depicts some examples of p-optimality and completeness. Clearly, a path is p-optimal only if a superset of $\T$ is among its extensions. Moreover, path 1 reveals that the optimal and p-optimal path notions are equal when $n_f = 0$.

\begin{figure*}[!t]
  \centering
\centerline{\includegraphics[width = 0.7\linewidth]{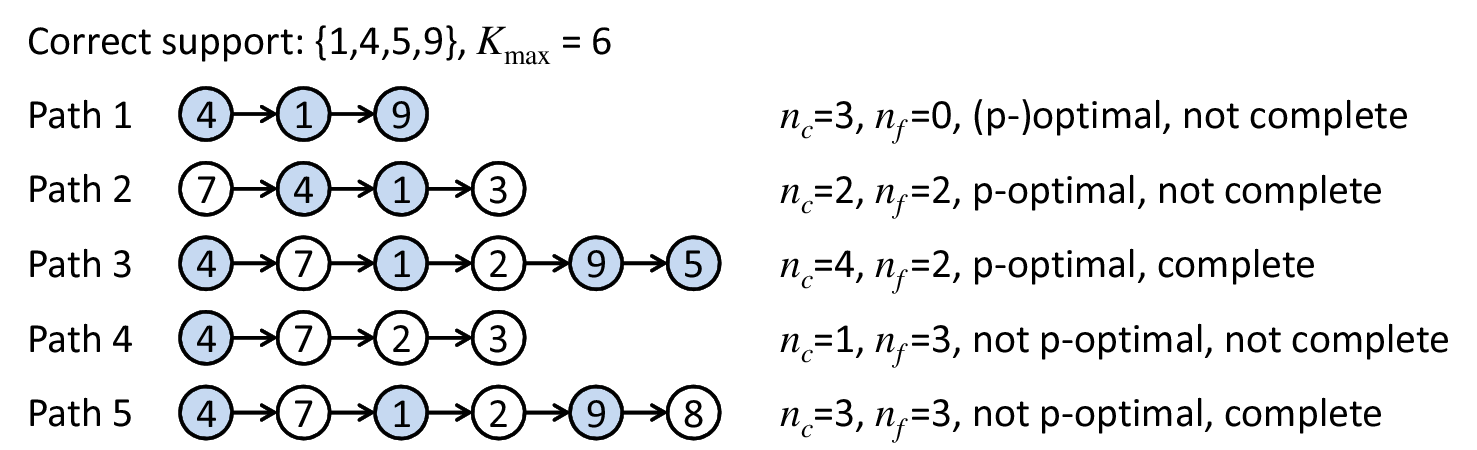}}
\caption{Completeness and p-optimality with respect to $n_c$ and $n_f$.}
\label{fig:optimality}
\end{figure*}

Next, we state the following lemma:

\begin{lemma} \label{lemma:ExpSuboptPath}
Let $i$ be a p-optimal path with $n_c$ correct and $n_f$ incorrect indices. If
\begin{equation}
\delta_{K+n_f+B} < \frac{\sqrt{B}}{\sqrt{K-n_c}+\sqrt{B}},
\label{eq_AOMPe_res_1}
\end{equation}
holds for path $i$, the set of its extensions by each of its best $B$ children contains at least one p-optimal path with $n_c+1$ correct indices. Moreover, (\ref{eq_AOMPe_res_1}) holds for all p-optimal paths in this set.
\end{lemma}
\begin{proof}
By Theorem~\ref{Thrm_Main}, expansion of path $i$ is successful when (\ref{eq_AOMPe_res_1}) holds. Hence, it introduces at least one p-optimal path, say $j$, with $n_c+1$ correct and $n_f$ incorrect indices. Moreover, the upper bounds from (\ref{eq_AOMPe_res_1}) are related as
\begin{equation}
   \frac{\sqrt{B}}{\sqrt{K-n_c}+\sqrt{B}} < \frac{\sqrt{B}}{\sqrt{K-n_c-1}+\sqrt{B}} \nonumber
\end{equation}
where the left and right sides are for the upper bounds on the RIC's corresponding to the number of correct and false nodes in path $i$ and $j$, respectively.
Since the upper bound is larger for path $j$, and (\ref{eq_AOMPe_res_1}) holds for path $i$, (\ref{eq_AOMPe_res_1}) should also hold for path $j$.
\end{proof}

Now an online recovery condition can be presented for {\AOMPe}:
\begin{theorem}
\label{Thrm_AOMPe} Set $\varepsilon=0$ and $\Kmax \leq M-K$. Let $\mathbf{\Phi}$ be full rank. Assume that pruning is p-optimal. Then, {\AOMPe} perfectly recovers a $K$-sparse signal from noise-free measurements if the search, at any step, expands a path which satisfies $K+n_f \leq \Kmax$ and
\begin{equation}
\delta_{K+n_f+B} < \frac{\sqrt{B}}{\sqrt{K-n_c}+\sqrt{B}}.
\label{eq_AOMPe_res1}
\end{equation}
\end{theorem}
\begin{proof}
As $K+n_f \leq \Kmax$, the best path at this step, $b$, is p-optimal.
Moreover,  Lemma~\ref{lemma:ExpSuboptPath} and (\ref{eq_AOMPe_res1}) guarantee p-optimality of at least one child of $b$.
By assumption pruning cannot remove all p-optimal paths.
Altogether, these guarantee the existence of at least one p-optimal path in the tree until termination.

On the other hand, the criterion $\varepsilon=0$ requires that the residue should vanish for termination. Since $\Kmax \leq M-K$ and $\mathbf{\Phi}$ is full rank, the residue may vanish if and only if the support estimate is a superset of $\T$. Therefore, the search must terminate at a p-optimal path containing the true support unless there remain no p-optimal paths in the search tree. As $\mathbf{\Phi}$ is full rank, the orthogonal projection of $\y$ onto this set yields exactly $\x$. Combining with the existence of at least one p-optimal path, exact recovery of $\x$ is guaranteed.
\end{proof}

Note that the condition $\varepsilon = 0$, stated in Theorem~\ref{Thrm_AOMPe} for the sake of theoretical correctness, translates into a very small $\varepsilon$ in practice to account for the numerical computation errors.

Since  Theorem~\ref{Thrm_AOMPe} depends on the existence of a p-optimal path satisfying (\ref{eq_AOMPe_res1}), it does not provide overall guarantees for all $K$-sparse signals as Theorem~\ref{Thrm_AOMPK} does.
In contrast, Theorem~\ref{Thrm_AOMPe} implies online guarantees depending on the support estimates of the intermediate (i.e., neither complete nor empty) p-optimal paths.
In fact, it is hardly possible to provide guarantees for the existence of such paths.
Nonetheless, Theorem~\ref{Thrm_AOMPe} can be generalized starting with the empty path:
\begin{theorem}
\label{Thrm_AOMPe_2} Set $\varepsilon=0$, $I \geq B$ and $\Kmax \leq M-K$. Let $\mathbf{\Phi}$ be full rank. Assume pruning is optimal. Then, {\AOMPe} perfectly recovers all $K$-sparse signals from noise-free measurements if $\mathbf{\Phi}$ satisfies RIP with
\begin{equation}
\delta_{K+B} < \frac{\sqrt{B}}{\sqrt{K}+\sqrt{B}}.
\label{Eq:AOMPe_res2}
\end{equation}
\end{theorem}

We omit the proof of Theorem~\ref{Thrm_AOMPe_2}, since it follows trivially from Theorem~\ref{Thrm_AOMPe} by replacing $n_f=n_c=0$.

Theorem~\ref{Thrm_AOMPe_2} provides overall guarantees for all sparse signals as Theorem~\ref{Thrm_AOMPK}.
We observe that both theorems require the same RIP condition for exact recovery of all sparse signals.

Although Theorem~\ref{Thrm_AOMPK} and Theorem~\ref{Thrm_AOMPe_2} provide similar overall guarantees, {\AOMPe} also possesses the online guarantees of Theorem~\ref{Thrm_AOMPe}.
Section~\ref{Sec:Analysis_Comp} presents an analytical comparison of the conditions in Theorem~\ref{Thrm_AOMPe} and Theorem~\ref{Thrm_AOMPe_2}.
This study states that Theorem~\ref{Thrm_AOMPe} may be satisfied even when Theorem~\ref{Thrm_AOMPe_2} fails.
This reveals the importance of the online guarantees to comprehend the improved recovery accuracy of {\AOMPe}.

\subsection{A Note on Pruning}
\label{sec:Pruning}

Theorem~\ref{Thrm_AOMPK} and Theorem~\ref{Thrm_AOMPe_2} both rely on optimal/p-optimal pruning, which is hard to prove analytically.
Though this may be seen as a limitation of the theoretical findings, it is  obvious that pruning is unavoidable for the tractability of the search.
As an empirical justification, it is important to observe that the true solution may be reached along different paths in the tree.
This is due to the fact that the ordering of the nodes is not important. Hence, the tree is subject to contain a large number of candidate solutions (optimal/p-optimal paths), while it is  enough for exact recovery when only one of these optimal paths is not pruned.

Moreover, the theoretical analysis of MMP in \cite{Kwon:MMP} is also subject to an equivalent assumption in practice.
Though the authors analyse MMP without any limit on the number of explored paths, they acknowledge that this is impractical.
Similar to {\AOMP}, they limit the number of paths for the empirical evaluation of MMP.
It is clear that the theoretical findings of \cite{Kwon:MMP}, which do not address this pruning strategy, are practically meaningful only with an assumption on pruning.

\subsection{On the Validity of the Online Guarantees}
\label{Sec:Analysis_Comp}

To address the validity of the online condition in Theorem~\ref{Thrm_AOMPe}, we show that it can be satisfied when the overall guarantees in Theorem~\ref{Thrm_AOMPe_2} fail.
The next theorem reveals that a p-optimal path satisfying (\ref{eq_AOMPe_res1}) may be found even when (\ref{Eq:AOMPe_res2}) fails.

\begin{theorem}
\label{Thrm_AOMP_Comp}
Assume $K\geq(3+2\sqrt{B})^2$. If $1\leq n_f+B\leq\lceil K/2\rceil$ and $n_c$ satisfies (\ref{Eq:n_cBound}) at some intermediate iteration, (\ref{eq_AOMPe_res1}) becomes less restrictive than (\ref{Eq:AOMPe_res2}).
\end{theorem}
\begin{proof} Assume that
\begin{equation}
\label{Eq_Thrm_OMP2_1}
\delta_{K+n_f+B} \geq \frac{3\sqrt{B}}{\sqrt{K}+\sqrt{B}}.
\end{equation}
Since $n_f + B \leq \lceil K/2 \rceil$, we can write $3\lceil K/2 \rceil \geq K+n_f+B$. By Lemma~\ref{lemma_monotonicity}, we obtain
\begin{equation}
\delta _{3\lceil K/2 \rceil} \geq \frac{3\sqrt{B}}{\sqrt{K}+\sqrt{B}}. \nonumber
\end{equation}
Then, Lemma~\ref{lemma_ExRec2} yields
\begin{equation}
\delta_{K+B} > \frac{\sqrt{B}}{\sqrt{K}+\sqrt{B}} \nonumber
\end{equation}
which clearly contradicts (\ref{Eq:AOMPK_res1}).
In contrast, Lemma~\ref{lemma_ExRec1} yields
\begin{equation}
 \frac{3\sqrt{B}}{\sqrt{K}+\sqrt{B}} \leq \frac{\sqrt{B}}{\sqrt{K-n_c}+\sqrt{B}} \nonumber
\end{equation}
for $n_c$ satisfying (\ref{Eq:n_cBound}) when $K \geq (3+2\sqrt{B})^2$. That is, there exists some range of $\delta_{K+n_f+B}$ such that
\begin{equation}
\label{Eq_OMP_1}
 \frac{3\sqrt{B}}{\sqrt{K}+\sqrt{B}} \leq \delta_{K+n_f+B} \leq \frac{\sqrt{B}}{\sqrt{K-n_c}+\sqrt{B}}. \nonumber
\end{equation}
Hence, under the given conditions, there exists some $\delta_{K+n_f+B}$ satisfying (\ref{eq_AOMPe_res1}), but no $\delta_{K+B}$ satisfying (\ref{Eq:AOMPe_res2}).
\end{proof}

Theorem~\ref{Thrm_AOMP_Comp} clarifies that Theorem~\ref{Thrm_AOMPe} may hold even when Theorem~\ref{Thrm_AOMPe_2} fails.
In other words, {\AOMPe} possesses online guarantees for some sparse signals for which the overall guarantees fail.
This explains why {\AOMPe} improves the recovery accuracy over {\AOMPK}, and reveals that the residue-based termination is more optimal for noise-free sparse signal recovery than its sparsity-based counterpart.

Note that Theorem~\ref{Thrm_AOMP_Comp} is based on some $n_f$, $n_c$ and $K$ ranges for which it is provable. This is enough, since even a single supporting case already establishes the validity of Theorem~\ref{Thrm_AOMPe}.
On the other hand, we expect it to be valid for a wider range. This intuition is also supported by the following simulations, where {\AOMPe} improves recovery in almost all cases.

\section{Empirical Analyses}
\label{sec:results}

We demonstrate {\AOMP} in comparison to BP \cite{Chen:BP}, SP \cite{Dai:SP}, OMP \cite{Pati:OMP}, ISD \cite{Wang:ISD}, SL0 \cite{Mohimani:SL0}, IHT \cite{Blumensath:IHT2}, FBP \cite{FBP_DSP} and MMP-DF \cite{Kwon:MMP} in various scenarios involving synthetically generated signals and images. We mainly concentrate on comparison with different algorithms. Regarding the impact of {\AOMP} parameters such as $B$, $P$ and $\alpha$ on the recovery performance, we refer the reader to \cite{Karahanoglu:AOMPfull}, where the matter has been discussed with detailed empirical analysis. The numerical choices of the parameters in this work are mainly based on these findings, which we do not repeat here.

\subsection{Experimental Setup}

Unless given explicitly, the setup is as follows: We set $I=3$, $B=2$ and $P=200$. For {\AOMPe}, $\varepsilon$ is set to $10^{-6}$. This $\varepsilon$ is shared by OMP, which also runs until $\|\res\|_2 \leq \varepsilon\|\y\|_2$.
We select $\alpha_{\text{Mul}}=0.8$ for Mul-{\AOMPK}, $\alpha_{\text{Mul}}=0.9$ for Mul-{\AOMPe} and $\alpha_{\text{AMul}}=0.97$ for AMul-{\AOMPe}. We employ FBP with $\alpha=0.2M$ and $\beta = \alpha-1$ as suggested in \cite{FBP_DSP}.
For MMP-DF, we set the branching factor $L=6$ following \cite{Kwon:MMP}, and allow a maximum of 200 paths for a fair comparison with the {\AOMP} variants where $P=200$.
Each test is repeated over a randomly generated set of $S$ sparse samples.
For each sample, $\mathbf{\Phi}$ is drawn from the Gaussian distribution with mean zero and standard deviation $1/N$.
The nonzero entries of the test samples are selected from three random ensembles. The nonzero entries of the Gaussian sparse signals follow standard Gaussian distribution while those of the uniform sparse signals are distributed uniformly in $[-1,1]$.
The constant amplitude random sign (CARS) sparse signals have unit magnitude nonzero elements with random sign.
The average normalized mean-squared-error (ANMSE) is defined as
\begin{equation}\label{Eq:ANMSE}
    \text{ANMSE} = \frac{1}{S} \sum_{i=1}^{S}{\frac{\|\x_i - \hat{\x}_i\|_2^2}{\|\x_i\|_2^2}}
\end{equation}
where $\hat{\x}_i$ is the recovery of the $i$th test vector $\x_i$.

We perform the {\AOMP} recovery using the AStarOMP software\footnote{Available at http://myweb.sabanciuniv.edu/karahanoglu/research/.}. AStarOMP implements the A$^\star$ search by an efficient trie\footnote{Trie is an ordered tree data structure in computer science. The ordering provides important advantages for the implementation of {\AOMP}, such as reducing the cost of addition/deletion of paths and finding equivalent paths.} structure \cite{Sedgewick:Trie}, where the nodes are ordered with priorities proportional to their inner products with $\y$. This maximizes the number of shared nodes between paths, allowing a more compact tree representation and faster tree modifications. The orthogonal projection is performed via the QR factorization.

\subsection{Exact Recovery Rates and Reconstruction Error}

\begin{figure*}[!t]
  \centering
\centerline{\includegraphics[width=0.9\linewidth]{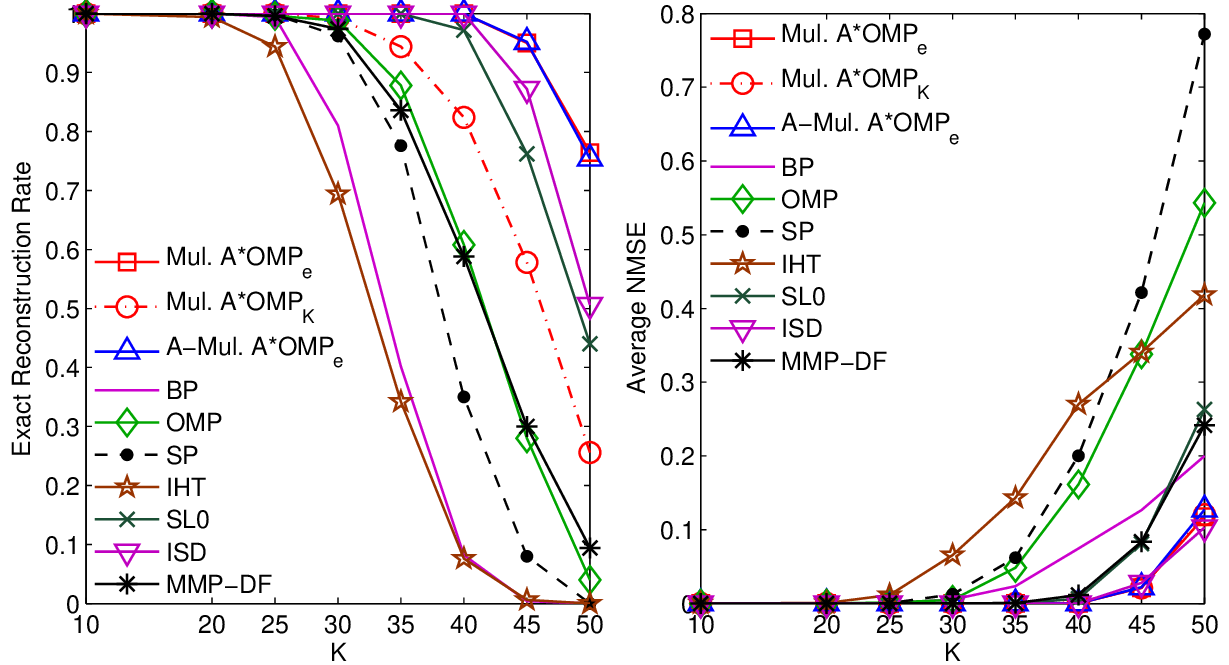}}
\includegraphics[width=0.45\linewidth]{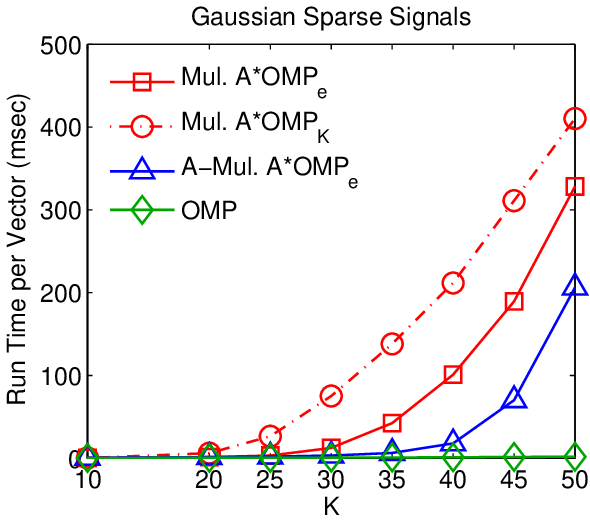}
\caption{Recovery results and average run time for the Gaussian sparse signals.}
\label{fig:gauss}
\end{figure*}

\begin{figure*}[!t]
\begin{center}
\includegraphics[width=\linewidth]{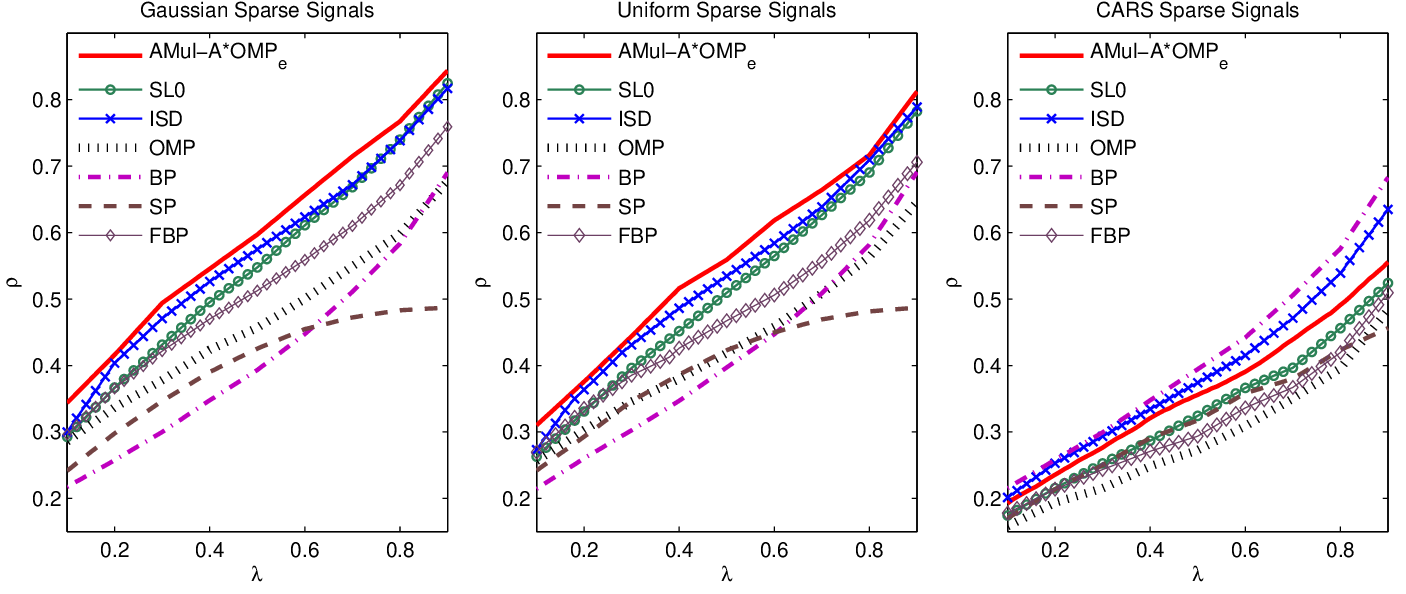}
\caption{Phase transitions of AMul-{\AOMPe}, BP, SP, OMP, ISD and SL0 for Gaussian, uniform and CARS sparse signals.}
\label{fig:PT}
\end{center}
\end{figure*}

The first simulation deals with the recovery of Gaussian sparse signals, where $N=256$, $M=100$, $K\in[10,50]$, $\Kmax=55$ and $S=500$.
The results are depicted in \figurename~\ref{fig:gauss}, which reveals that {\AOMP} performs significantly better than the other algorithms.
{\AOMPe} provides exact recovery until $K=40$, which is clearly the best.
We observe similar ANMSE among {\AOMP} variants, while the residue-based termination improves the exact recovery rates significantly. Evidently, {\AOMPe} is better than {\AOMPK} at identifying smaller magnitude entries, which hardly change the ANMSE, however increase the exact recovery rates.
It is also important that all {\AOMP} variants yield higher recovery rates than MMP-DF. Especially, Mul-{\AOMPK} and MMP-DF are interesting to compare since both return solutions with $K$ indices. We observe that Mul-{\AOMPK} yields significantly higher recovery rates than MMP-DF, which implies the effectiveness of the sophisticated search techniques employed in {\AOMP}.

As for the average run times\footnote{OMP and {\AOMP} are tested using the AStarOMP software.
The other algorithms are ignored as they run in MATLAB, which is slower.}, both the residue-based termination and the AMul cost model significantly accelerate {\AOMP} due to the relaxation of $\alpha$ to larger values.
Since AMul-{\AOMPe} can afford the largest $\alpha$, it is the fastest {\AOMP} variant.
This confirms the claim in Section~\ref{sec:AStar} that increasing $\alpha$ reduces the number of explored nodes and accelerates \AOMP.

In addition to this example, \cite{AOMP_EUSIPCO} contains simulations for uniform and binary sparse signals and noisy measurements.
These simulations indicate that AMul-{\AOMPe} improves the recovery for uniform sparse signals and noisy cases as well.

\subsection{Phase Transitions}

Empirical phase transitions provide important means for recovery analysis, since they reveal the recovery performance over the feasible range of $M$ and $K$. Consider the normalized measures $\lambda = M/N$ and $\rho = K/M$. The phase transition curve is mostly a function of $\lambda$ \cite{Maleki:TST}, hence it allows for a general characterization of the recovery performance.

To obtain the phase transitions, we fix $N=250$, and alter $M$ and $K$ to sample the $\{\lambda, \rho\}$ space for $\lambda \in [0.1, 0.9]$ and $\rho \in [0,1]$. For each $\{\lambda,\rho\}$ tuple, we randomly generate 200 sparse instances and perform the recovery.
Setting the exact recovery criterion as $\frac{\|\x_i - \hat{\x}_i\|_2}{\|\x_i\|_2} \leq 10^{-2}$, we count the exactly recovered samples. The phase transitions are then obtained as in \cite{Maleki:TST}. For each $\lambda$, we employ a generalized linear model with logistic link to describe the exact recovery curve over $\rho$, and then find $\rho$ which yields $50\%$ exact recovery probability. Combination of these $\rho$ values gives the phase transition curve.

Let us first concentrate on $K_\text{max}$ and define the normalized measure $\rho_\text{max} = K_\text{max}\mydiv M$. Once we identify the optimal $\rho_\text{max}$ over $\lambda$, we can set $K_\text{max}$ for particular $\lambda$ and $M$.
To find the optimal $\rho_\text{max}$, we have run a number of simulations and observed that the recovery performance of AMul-{\AOMPe} is quite robust to $\rho_\text{max}$, with a perturbation up to $3\%$. Hence, the recovery accuracy is mostly independent of $\rho_\text{max}$. Yet, based on our experience, we set $\rho_\text{max} = 0.5 + 0.5\lambda$ taking into account both the accuracy and complexity of the search.

\figurename~\ref{fig:PT} depicts the phase transition curves. Clearly, AMul-{\AOMPe} yields better phase transitions than the other algorithms for the Gaussian and uniform sparse signals.
We also observe that FBP provides fast approximations with better accuracy than BP and the two other greedy competitors, SP and OMP for these two cases. This reveals that FBP is actually suitable to applications where speed is crucial.
On the other hand, BP and ISD are the best performers for the CARS case, while AMul-{\AOMPe} is the third best.
We observe that BP is robust to the coefficient distribution, while the phase transitions for AMul-{\AOMPe} and OMP exhibit the highest variation among distributions.
This indicates that OMP-type algorithms are more effective when the nonzero elements span a wide magnitude range such as the Gaussian distribution. Moreover, if this range gets wide enough, even OMP can outperform BP.
In parallel, the CARS ensemble is referred to as the most challenging case for the greedy algorithms in the literature \cite{Dai:SP, Maleki:TST}.
This can be understood analytically by considering the span of the correlation between $\mathbf{\Phi}_\T$ and $\y$. The detailed analytical analysis in \cite{Karahanoglu:phdthesis} state that this span gets smaller when the magnitudes of the nonzero elements get closer, and vice versa. When this span gets smaller (for the CARS ensemble in the limit), wrong indices are more likely to be selected by OMP, increasing the failure rate.

\subsection{A Hybrid Approach for Faster Practical Recovery}

\begin{figure*}[!t]
\begin{center}
\includegraphics[width=0.8\linewidth]{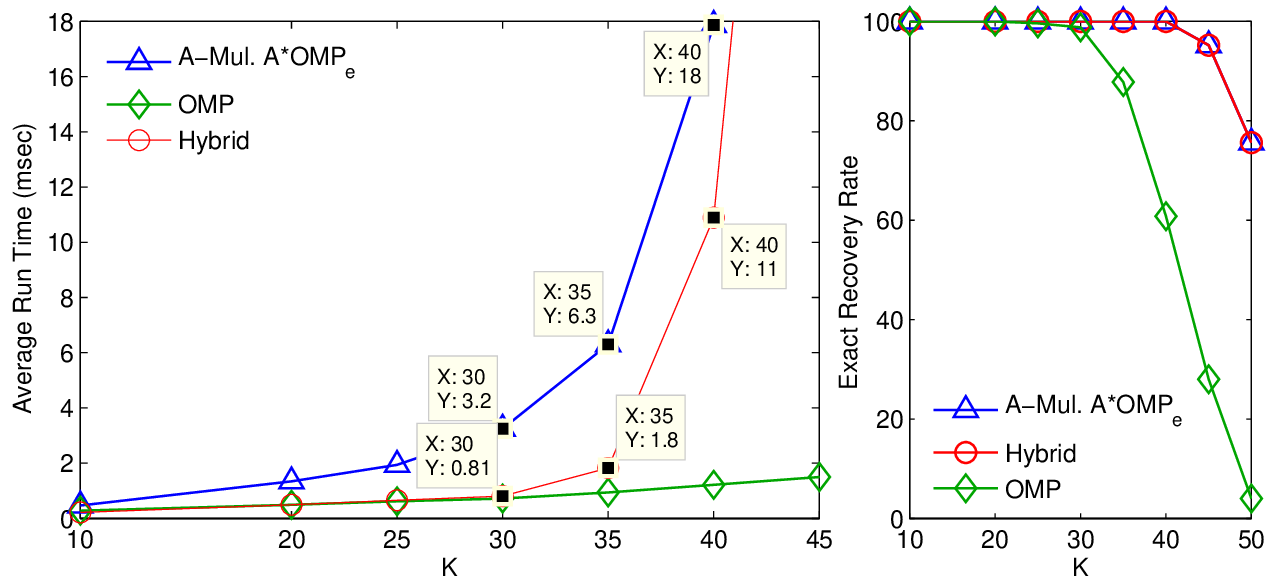}
\end{center}
\caption{Performance of the hybrid scheme for the Gaussian sparse vectors.}
\label{fig:Hybrid}
\end{figure*}

We have observed that OMP provides exact recovery up to some mid-sparsity range.  Moreover, there are regions where AMul-{\AOMPe} provides exact recovery while OMP recovery rates are also quite high.
In these regions, we can accelerate the recovery without sacrificing the accuracy by a two-stage hybrid scheme.
We first run OMP, and then AMul-{\AOMPe} only if OMP fails.
Assuming that ${K{+}\Kmax}$-RIP holds, a non-vanishing residue indicates OMP failure, and then AMul-{\AOMPe} is run. This reduces the number of AMul-{\AOMPe} runs and accelerates the recovery.
Moreover, we use the order by which OMP chooses the vectors for setting the priorities of trie nodes. A vector OMP chooses first gets higher priority, and is placed at lower levels of the trie. This reduces not only the trie size but also trie modification costs.

According to the recovery results in \figurename~\ref{fig:Hybrid}, AMul-{\AOMPe} and the hybrid approach yield identical exact recovery rates, while the latter is significantly faster. This acceleration is proportional to the exact recovery rate of OMP. That is, the hybrid approach is faster where OMP is better. These results show that this approach is indeed able to detect the OMP failures, and run AMul-{\AOMPe} only for those instances.

\subsection{Demonstration on a Sparse Image}

To illustrate AMul-{\AOMPe} on a more realistic coefficient distribution, we demonstrate recovery of some commonly used $512 \times 512$ images including ‘Lena’, ‘Tracy’, ‘Cameraman’, etc. The recovery is performed in $8 \times 8$ blocks in order to break the problem into smaller and simpler subproblems. Each image is first preprocessed to obtain $K$-sparse blocks in the 2D Haar Wavelet basis $\mathbf{\Psi}$ by keeping the $K$ largest magnitude wavelet coefficients for each block\footnote{This formulation involves the structured dictionary $\mathbf{\Psi}$ for sparse representation of the images. The observation model becomes $\y =  \mathbf{\Phi}\mathbf{\Psi}\x$, where the reconstruction basis is not $\mathbf{\Phi}$ alone, but $\mathbf{\Phi}\mathbf{\Psi}$.}. We select $K=12$ for the image 'Bridge', and $K=14$ for the rest. $M=32$ observations are taken from each block. The entries of $\mathbf{\Phi}$ are randomly drawn from the Gaussian distribution with mean 0 and standard deviation $1/N$. We set $I=3$, $P=200$ and $\Kmax = 20.$ $\alpha_{\text{AMul}}$ is reduced to 0.85 in order to compensate the decrement in the auxiliary term of (\ref{Eq:AMul_CM}) due to smaller $\Kmax$. Peak Signal-to-Noise Ratio (PSNR) values obtained by employing different recovery algorithms are given in Table~\ref{Table:compImage}.
For each image, maximum and minimum PSNR values obtained are shown in bold and in italics, respectively.
Mean PSNR for each algorithm is also given on the last row.
AMul-{\AOMPe} exhibits significant improvements over the competitors and yields the maximum PSNR for all images tested.
Increasing $B$ from 2 to 3 further improves PSNR, yielding improvements of 15.7 dB over BP, 15.3 dB over ISD, 17.7 dB over SL0, and 11.7 dB over MMP-DF on the average. As a visual example, we depict  the reconstruction of the test  image ``Bridge'' using BP and AMul-{\AOMPe} with $B=3$ in \figurename~\ref{fig:bridge}. In this case, BP yields 29.9 dB PSNR, while AMul-{\AOMPe} improves the PSNR to 51.4 dB. Though not depicted in \figurename~\ref{fig:bridge}, AMul-{\AOMPe} outperforms BP with 46.8 dB when $B=2$ as well. A detailed investigation of the recovered images reveals that AMul-{\AOMPe} improves the recovery especially at detailed regions and boundaries.

\renewcommand{\arraystretch}{1.2}
\begin{table*}[!t]
\centering
\caption{PSNR values for images reconstructed using different algorithms. Maximum and minimum PSNR values are shown in bold and in italics, respectively. Mean PSNR values over the whole set of images are given in the last row of the table.}
{\footnotesize{
\begin{tabular}{p{2cm}  x{1.3cm}  x{1.3cm}  x{1.3cm}  x{1.3cm} x{1.3cm}  x{1.3cm}  x{1.3cm}  x{1.3cm}  x{1.3cm}}  \hline \hline
&\multirow{2}{*}{BP}&\multirow{2}{*}{OMP}&\multirow{2}{*}{SP} &\multirow{2}{*}{IHT}&\multirow{2}{*}{ISD}&\multirow{2}{*}{SL0} & \multirow{2}{*}{MMP-DF} & \multicolumn{2}{c}{AMul-A*OMP}  \tabularnewline \cline{9-10}
&&&&&&&&B=2&B=3 \tabularnewline \hline \hline
Bridge	  	& 29.9 & 26.9 & 24.6 & \emph{19.8} & 32.1 & 29.1 & 36.9 &46.8  & \textbf{51.4} 	\tabularnewline \hline
Lena      	& 33.5 & 29.6 & 27.5 & \emph{22.9} & 33.2 & 30.6 & 36.4 & 42.6 & \textbf{47.1}  \tabularnewline \hline
Tracy    	& 40.6 & 36.8 & 33.9 & \emph{27.6} & 39.4 & 38.2 & 43.8 & 52.5 & \textbf{56.8} 	\tabularnewline \hline
Pirate   	& 31.7 & 27.7 & 25.3 & \emph{21.5} & 32.4 & 30.3 & 35.7 & 40.3 & \textbf{43.4} 	\tabularnewline \hline
Cameraman 	& 34.4 & 30.7 & 28.5 & \emph{23}    & 33.9 & 31.8 & 37.4 & 48.3 & \textbf{54.7}  	\tabularnewline \hline
Mandrill  	& 28.3 & 24.4 & 22.1 & \emph{19.2} & 29.7 & 26.8 & 32.1 & 36.3 & \textbf{39.9} 	\tabularnewline \hline
Mean PSNR  	& 33.1 & 29.4 & 27    & \emph{22.3} & 33.5 & 31.1 & 37.1 & 44.5 & \textbf{48.8}	\tabularnewline \hline \hline
\end{tabular}
}}
\label{Table:compImage}
\end{table*}

\begin{figure*}[!t]
\begin{center}
\includegraphics[width=\linewidth]{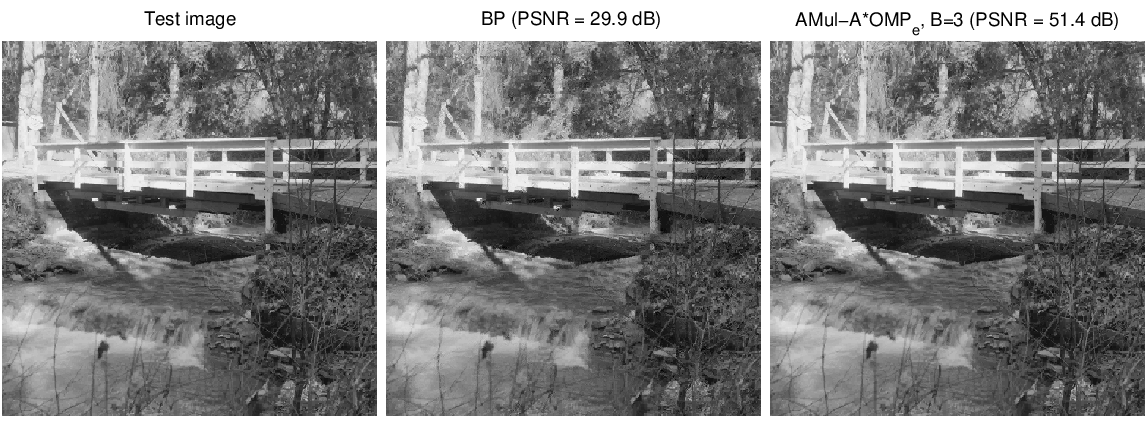}
\end{center}
\caption{Recovery of the image ``Bridge" using BP and AMul-{\AOMPe}.}
\label{fig:bridge}
\end{figure*}

\section{Summary}
\label{sec:conc}

The fundamental goal of this manuscript is a comprehensive analysis of sparse recovery using {\AOMP}, with a particular focus on the novel variant AMul-{\AOMPe}. We have addressed this issue with emphasis on both theoretical and practical aspects.

We have presented a theoretical analysis of signal recovery with {\AOMP}. We have first derived a RIP condition for the success of an {\AOMP} iteration. Then, we have generalized this result for the exact recovery of all $K$-sparse signals from noise-free measurements both with {\AOMPK}, where the termination is based on the sparsity level $K$, and with {\AOMPe}, which employs the residue-based termination criterion. We have observed that both {\AOMP} variants enjoy similar RIP-based general exact recovery guarantees.
In addition, we have presented online guarantees for {\AOMPe}, which can be satisfied even when the general guarantees fail.
This has led to the conclusion that {\AOMPe} is more advantageous for sparse recovery, since it enjoys some online guarantees in addition to the general recovery guarantees analogous to those of {\AOMPK}.
This encourages utilising the residue-based termination instead of the sparsity-based one for recovery from noise-free observations.

We have demonstrated AMul-{\AOMPe} in a wide range of recovery simulations involving sparse signals with different characteristics.
According to the results, {\AOMP} performs better recovery than all the other candidates for uniform and Gaussian sparse signals.
Among {\AOMP} variants, AMul-{\AOMPe} promises the most accurate recovery and fastest execution times.
For CARS sparse signals, AMul-{\AOMPe} recovery is still better than the involved greedy alternatives, however BP is the most accurate algorithm in this case.
We have also shown that the search can be significantly accelerated without sacrificing the accuracy via a hybrid approach, which first applies OMP, and then AMul-{{\AOMPe}} only if OMP fails.
Finally, we have employed AMul-{\AOMPe} on sparse images, where it improves the recovery significantly over BP.

The AMul cost model with the residue-based termination has demonstrated strong empirical performance while also providing more greed due to the allowance for a larger $\alpha$. Hence, AMul-{\AOMPe} turns out to be the most promising {\AOMP} variant in this manuscript.
As future work, it is worth to investigate different cost model structures which may improve speed and convergence of the algorithm in specific problems.
For example, the cost model may be formulated to reflect the expected statistics of the signal of interest. Such a strategy would be problem-dependent, however it may guide the algorithm faster and more accurately to the desired solution.
Combining {\AOMP} with sparsity models, i.e., signals with specific sparsity patterns, is another promising future work direction.

\bibliographystyle{model1-num-names}

\vspace{1cm}

Nazim Burak Karahanoglu is a senior researcher at the Informatics and Information Security Research Center, The Scientific and Technological Research Council of Turkey (TUBITAK BILGEM) in Kocaeli, Turkey. He received his B.S. degree in Electrical and Electronics Engineering from METU, Ankara  in 2003, M.S. degree in Computational Engineering from the Friedrich-Alexander University of Erlangen-Nuremberg, Germany in 2006 and Ph.D. degree in Electronics Engineering from Sabanci University, Istanbul in 2013. He has been with TUBITAK BILGEM since 2008.  His research interests include compressed sensing and sonar signal processing.

\vspace{1cm}
Hakan Erdogan is an associate  professor at Sabanci University in Istanbul, Turkey. He received his B.S. degree in Electrical Engineering and Mathematics in 1993 from METU, Ankara and his M.S. and Ph.D. degrees in Electrical Engineering: Systems from the University of Michigan, Ann Arbor in 1995 and 1999 respectively. He was with the Human Language Technologies group at IBM T.J. Watson Research Center, NY between 1999 and 2002. He has been with Sabanci University since 2002. His research interests are in developing and applying probabilistic methods and algorithms for multimedia information extraction.

\end{document}